\PassOptionsToPackage{english}{babel}
\documentclass[12pt]{article}

\usepackage[english]{babel} 
\usepackage{setspace}
\usepackage[margin=2.5cm]{geometry}
\usepackage{amsthm, amssymb, amstext}
\usepackage[fleqn]{amsmath}
\usepackage{latexsym}
\usepackage[dvips]{graphicx}
\usepackage{comment}
\usepackage{mathtools}
\usepackage{enumerate}

\usepackage{todonotes}
\usepackage{comment}

\newcommand{\bigOh}[1]{\mathcal{O}\mathopen{}\mathclose\bgroup\left( #1 \aftergroup\egroup\right)}
\newcommand{\bigO}[1]{\mathcal{O}(#1)}

\newtheorem{theorem}{Theorem}
\newtheorem{lemma}{Lemma}

\newtheorem{corollary}{Corollary}[section]
\newtheorem{conjecture}{Conjecture}

\newtheorem{claim}{Claim}
\newtheorem{problem}{Problem}

\theoremstyle{definition}

\usepackage{amsthm}

\title{\textbf{Towards Cereceda's conjecture \\for planar graphs}}
\author{Eduard Eiben \quad Carl Feghali  \\ Department of Informatics, \\ University of Bergen, \\ Bergen, Norway \\ {\small \texttt{eduard.eiben@uib.no} \quad \texttt{carl.feghali@uib.no}}}
\date{}

\begin{document}   

\maketitle

\begin{abstract}
The reconfiguration graph $R_k(G)$ of the $k$-colourings
of a graph~$G$ has as vertex set the set of all possible $k$-colourings
of $G$ and two colourings are adjacent if they differ on the colour of exactly
one vertex. 

Cereceda conjectured ten years ago that, for every $k$-degenerate graph $G$ on $n$ vertices, $R_{k+2}(G)$ has diameter $\bigO{n^2}$. 
The conjecture is wide open, with a best known bound of $\bigO{k^n}$, even for planar graphs.
We improve this bound for planar graphs to $2^{\bigO{\sqrt{n}}}$. Our proof can be transformed into an algorithm that runs in $2^{\bigO{\sqrt{n}}}$ time. \end{abstract}

\section{Introduction}

Let $G$ be a graph, and let $k$ be a non-negative integer. 
A $k$-colouring of $G$ is a function $f: V(G) \rightarrow \{1, \dots, k\}$ such that $f(u) \not= f(v)$ whenever $(u, v) \in E(G)$. 
The reconfiguration graph $R_k(G)$ of the $k$-colourings of $G$ has as vertex set the set of all $k$-colourings of $G$ and two vertices of $R_k(G)$ 
are adjacent if they differ on the colour of exactly one vertex. Let $d$ be a positive integer. Then $G$ is said to be $d$-degenerate 
if every subgraph of $G$ contains a vertex of degree at most $d$. Expressed differently, $G$ is $d$-degenerate if there is an ordering $v_1, \dots, v_n$ of its vertices such that $v_i$ has at most $d$ neighbours $v_j$ with $j < i$.  

In the past decade, the study of reconfiguration graphs for graph colourings has been the subject of much attention. One typically asks whether the reconfiguration graph is connected. If so, what is its diameter and, in case it is not, what is the diameter of its connected components? See \cite{BJLPP14, CHJ06, CHJ06a} for some examples. Computational work has focused on deciding whether there is a path in the reconfiguration graph between a given pair of colourings \cite{BC09, CHJ06b, JKKPP14}. Other structural considerations of the reconfiguration graph have also been investigated in \cite{asplund, beier}. Reconfiguration graphs have also been studied for many other decision problems; see \cite{nishimura} for a recent survey. 

We remark that reconfiguration problems for graph colourings do not have known results for which the reconfiguration graph is connected but has a diameter that is not polynomial in the order of the graph. (In nearly all cases, the diameter turns out to be quadratic in the number of vertices.) On the other hand, the problem of deciding whether a pair of colourings are in the same component of the reconfiguration graph tends to be $\textsc{PSPACE}$-complete whenever the reconfiguration graph is disconnected. There are exceptions to this pattern such as, for example, deciding whether a pair of 3-colourings of a graph belong to the same component \cite{CHJ06b}. 

Given  a $d$-degenerate graph $G$, it is not difficult to show that $R_{d+2}(G)$ is connected \cite{jerrum}. The foregoing pattern motivated Cereceda \cite{luisthesis} to  conjecture that $R_{d+2}(G)$ has diameter that is quadratic in the order of $G$.

\begin{conjecture}\label{conj:cereceda}
Let $d$ be a positive integer, and let $G$ be a $d$-degenerate graph on $n$ vertices. Then $R_{d+2}(G)$ has diameter $\bigO{n^2}$.
\end{conjecture}

Conjecture \ref{conj:cereceda} has  resisted several efforts and has only been verified (other than for trees) for graphs of bounded treewidth \cite{bonamy13} and graphs with degeneracy at least $\Delta - 1$ where $\Delta$ denotes the maximum degree of the graph \cite{feghalibrooks}.

In the expectation of the difficulty of Conjecture \ref{conj:cereceda}, 
 Bousquet and Perarnau \cite{bousquet11} have shown that for every $d\geq 1 $ and $\epsilon > 0$ and every graph $G$ with maximum average degree $d - \epsilon$, the diameter of $R_{d + 1}(G)$ is $\bigO{n^c}$ for some constant $c = c(d, \epsilon)$; see \cite{feghali} for a short proof.  Their result in particular implied that the reconfiguration graph of $8$-colourings for planar graphs has diameter that is polynomial in the order of the graph.  Since planar graphs are $5$-degenerate, the one outstanding case of Conjecture \ref{conj:cereceda} restricted to planar graphs is thus $k = 7$ (aside, of course, from improving the constant term in the exponent of the diameter); see also \cite[Conjecture 16]{bonamy13}. 
On the other hand, the best known upper bound on the diameter in Conjecture \ref{conj:cereceda} is $\bigO{d^n}$ -- even for planar graphs -- and this follows from \cite{jerrum}. 
 In this note, we significantly improve this bound for planar graphs.

\begin{theorem}\label{thm:planar}
For every planar graph $G$ on $n$ vertices, $R_7(G)$ has diameter at most $2^{\bigO{\sqrt{n}}}$.
\end{theorem}

\section{Proof of Theorem \ref{thm:planar}} 

In this section, we prove Theorem \ref{thm:planar}. We begin with the following three lemmas. In the first lemma, we obtain a crude bound on the number of recolourings required for degenerate graphs to reduce the number of colours by one. 

\begin{lemma}\label{lem:algorithm}
Let $k \geq 1$, and let $G$ be a $k$-degenerate graph on $n$ vertices. Let $\{u_1, \dots, u_s\}$ be the set of vertices of $G$ of degree at least $k +2$. If $\alpha$ is a $(k+2)$-colouring of $G$, then we can recolour $\alpha$ to some $(k+1)$-colouring of $G$ by at most $
\bigO{n^2\prod_{i=1}^s \text{deg}(u_i)}$
recolourings. 
\end{lemma}

\begin{proof} 
We generalise the proof in \cite{feghalibrooks} by describing an algorithm that finds a sequence of recolourings from $\alpha$ to some $(k+1)$-colouring $\gamma$ of $G$ in $
\bigO{n^2\prod_{i=1}^s \text{deg}(u_i)}$
 time.

Let us fix a $k$-degenerate ordering $\sigma = v_1, \dots, v_n$ of $G$ and, without loss of generality, let $u_i$ appear before $u_j$ in $\sigma$ whenever $i<j$. 
In the following, we will describe an algorithm that given an index $h\in[n]$ outputs 
a sequence of recolourings with the following properties: 
\begin{itemize}
\item[(i)] for $i<h$, $v_i$ is not recoloured,
\item[(ii)] for $i> h$, $v_i$ is recoloured at most $\prod_{j=\ell}^s \text{deg}(u_j)$ times, where $u_\ell$ is the first degree at least $k+2$ vertex with index at least $h$ in $\sigma$, and
\item[(iii)] $v_h$ is recoloured once to a different colour.
\end{itemize}
Notice that the algorithm takes $\bigO{n\prod_{i=1}^s \text{deg}(u_i)}$ recolourings to recolour~$v_h$.
Hence, by repeatedly using such a sequence on the lowest index of a vertex coloured by the colour $k+2$, we can obtain the colouring $\gamma$ in which colour $k+2$ does not appear using $
\bigO{n^2\prod_{i=1}^s \text{deg}(u_i)}$ recolourings.

\newcommand{\recolor}{\textsc{recolour}}
Given $h\in [n]$ and a $k$-degenerate ordering $\sigma = v_1, \dots, v_n$ of $G$, the algorithm \recolor$(h)$ works as follows:

\begin{enumerate}
	\item If there is a colour $c$ that is not used on $v_h$ or any of its neighbours, then recolour $v_h$ to $c$ and terminate. 
	\item If $v_h$ has degree exactly $k+1$, let $v_{i}$ be the neighbour of $v_h$ that is latest in $\sigma$ and let $c$ be the colour of $v_i$. The algorithm first calls \recolor$(i)$ and then recolours $v_h$ to $c$. 
	\item If $v_h$ has degree at least $k+2$, then 
	\begin{enumerate}
		\item let $c$ be a colour not appearing on $v_h$ or any of its at most $k$ neighbours earlier in the ordering and
		\item let $v_{i_1},\ldots,v_{i_t}$ be the neighbours of $v_h$ later in the ordering with $i_1<i_2<\cdots<i_t$.
		\item For each $j\in[t]$ in the ascending order: if colour of $v_{i_j}$ is $c$ at this point, then call \recolor$(i_j)$.
		\item Recolour $v_h$ to $c$. 
	\end{enumerate}
\end{enumerate}

\noindent
We simultaneously prove the correctness and properties (i)--(iii) of the algorithm by induction on $n-h$. 
Clearly, whenever Step 1. of \recolor{} applies, we only recolour $v_h$ once to a colour not appearing on $v_h$ or any of its neighbours.
Moreover, if $h=n$, then, since $\sigma$ is a $k$-degenerate ordering, $v_h$ has degree at most $k$ and Step 1. again applies. 

For the induction step, let us assume that $h<n$, Step 1. of \recolor{} does not apply and for all $i>h$ the algorithm \recolor$(i)$ is correct and satisfies properties (i)--(iii). Since Step 1. does not apply, the degree of $v_h$ is at least $k+1$. We distinguish two cases. 

\begin{itemize}
	\item[Case 1:] $v_h$ has degree $k+1$. 
\end{itemize}
In this case \recolor{} applies Step 2. Since Step 1. does not apply, each colour appears either on $v_h$ or on one of its neighbours. Since there are $k+2$ colours, it follows that each colour appears precisely once in the closed neighbourhood of $v_h$. Note that since $\sigma$ is a $k$-degenerate ordering, the latest neighbour $v_i$ of $v_h$ is after $v_h$ in $\sigma$. Hence, properties (i) and (ii) follow from properties (i) and (ii) for \recolor$(i)$. Finally, the correctness and property (iii) follow from the fact that \recolor$(i)$ is, by induction, both correct and recolours the unique neighbour $v_i$ of $v_h$ of colour $c$ before recolouring $v_h$ to colour $c$. 

\begin{itemize}
	\item[Case 2:] $v_h$ has degree at least $k+2$. 
\end{itemize}
In this case \recolor{} applies Step 3. Again, the algorithm applies the recursive calls only on the vertices that are later in $\sigma$ than $v_h$ and hence property (i) is satisfied. After the execution of Steps (a)--(c), colour $c$ no longer appears on $v$ or any of its neighbours, after which $v_h$ is recoloured to $c$ at Step (d). So the algorithm is correct and property (iii) holds. To prove (ii), notice that the algorithm calls \recolor$(i)$ at most $\text{deg}(v_h)$ number of times and always for $i>h$. Hence, each vertex will get recoloured at most $\deg(v_h)\prod_{j=\ell+1}^s\deg(u_j)=\prod_{j=\ell}^s\deg(u_j)$, where $u_\ell=v_h$, and property (ii) follows. 
\end{proof}

The maximum average degree of a graph $G$ is defined as
\[
\operatorname{mad}(G) = \max\bigg\{\frac{\sum_{v \in V(H)} \deg(v)}{|V(H)|}: H \subseteq G\bigg\},
\]
By Euler's formula, the maximum average degree of a planar graph is strictly less than six.  By definition, if a graph has maximum average degree strictly less than $k$ for some positive integer $k$, then this graph is also $(k - 1)$-degenerate. 

In our next lemma, we show that we can reduce the number of colours by one using subexponentially many recolourings if we further assume our graph to have bounded maximum average degree.  

\begin{lemma}\label{lem:sub}
Let $k \geq 2$, and let $G$ be a a graph on $n$ vertices and with $\text{mad}(G) < k + 1$. If $\alpha$ is a $(k+2)$-colouring of $G$, then we can recolour $\alpha$ to some $(k+1)$-colouring of $G$ by $k^{\bigO{k^2\sqrt{n}}}$ recolourings.
\end{lemma}

\begin{proof}
We shall prove by induction on the size of $n=|V(G)|$ that we can recolour $\alpha$ to a $(k+1)$-colouring of $G$ such that each vertex in $G$ is recoloured at most $k^{\bigO{k^2\sqrt{n}}}$ times, which implies the lemma. We note that our inductive proof can be easily transformed to a recursive algorithm running in time $k^{\bigO{k^2\sqrt{n}}}$. 

As the base case, we show how to recolour by at most $k^{\bigO{k^2\sqrt{|H|}}}$ recolourings any graph $H$ with $\text{mad}(H)<k+1$ that contains at most $(k+1)2\sqrt{|H|}$ vertices of degree at most $k$.
Let $H$ be a graph with $\text{mad}(H)<k+1$ that contains at most $(k+1)2\sqrt{|H|}$ vertices of degree at most $k$,
let $\alpha^H$ be a $(k+2)$-colouring of $H$, and let $h = |V(H)|$.

\begin{claim} We can recolour $\alpha^H$ to some $(k+1)$-colouring of $H$ by $k^{\bigO{k^2\sqrt{h}}}$ recolourings. 
\end{claim}

\begin{proof}[Proof of Claim.]
	Due to Lemma~\ref{lem:algorithm}, 
	we only need to show that $\prod_{i = 1}^s\deg(u_i)\le k^{\bigO{k^2\sqrt{h}}}$, where $U = \{u_1, \dots, u_s\}$ is the set of vertices of degree at least $k+2$ in $H$. Let $W = \{w_1, \dots, w_t\}$ be the set of
	vertices of degree less than or equal to $k$ in $H$, and let $Z = V(H) \setminus (U \cup W) = \{z_1, \dots, z_{h - s - t}\}$ be the set of vertices of degree precisely $k + 1$ in $H$. 
	We can assume that $H$ is connected, because otherwise we can prove the claim for each connected component of~$H$.

Since $H$ is connected, 
\begin{eqnarray}\label{eq:in}
\bigg(\sum_{i=1}^s \text{deg}(u_i)\bigg) + \bigg(\sum_{i=1}^t \text{deg}(w_i)\bigg) &\geq& \bigg(\sum_{i=1}^s \text{deg}(u_i)\bigg) + t.
\end{eqnarray}
On the other hand, since $\text{mad}(G) < k + 1$, 
\begin{eqnarray}\label{eq:first}
& & \bigg(\sum_{i=1}^s \text{deg}(u_i)\bigg) + \bigg(\sum_{i=1}^t \text{deg}(w_i)\bigg) + \bigg(\sum_{i = 1}^{h - s - t} \text{deg}(z_i)\bigg)  < (k + 1)h \nonumber \\ 
&\Longleftrightarrow& \bigg(\sum_{i=1}^s \text{deg}(u_i)\bigg) + \bigg(\sum_{i=1}^t \text{deg}(w_i)\bigg) + (k + 1)(h - s - t) < (k + 1)(h - s - t + s + t) \nonumber \\
&\Longleftrightarrow& \bigg(\sum_{i=1}^s \text{deg}(u_i)\bigg) + \bigg(\sum_{i=1}^t \text{deg}(w_i)\bigg) < (k + 1)(s + t). 
\end{eqnarray}
Combine Inequalities (\ref{eq:in}) and (\ref{eq:first}):
\begin{eqnarray*}
 \bigg(\sum_{i=1}^s \text{deg}(u_i)\bigg) + t < (k + 1)(s + t)
  \Longrightarrow  (k + 2)s + t < (k + 1)(s + t)
 \Longleftrightarrow s < kt. 
\end{eqnarray*}
Since, by assumption, $t < (k+1)2\sqrt{h}$, it follows that $s < k(k+1)2\sqrt{h}$. Substituting these bounds into $(k + 1)(s + t) - t$ gives us \begin{eqnarray}\label{eq:second}
\sum_{i=1}^s \text{deg}(u_i)< k(k+1)^22\sqrt{h}+k(k+1)2\sqrt{h}\le 4k(k+1)^2\sqrt{h} = a.
\end{eqnarray}

By the AM-GM inequality of arithmetic and geometric means, it holds that
\begin{eqnarray}\label{eq:third}
\frac{\sum_{i=1}^s \deg(u_i)}{s} \geq \bigg(\prod_{i = 1}^s \deg(u_i)\bigg)^{s^{-1}}.
\end{eqnarray}
Combining Inequalities~(\ref{eq:second})~and~(\ref{eq:third}) we get 
\begin{eqnarray}
\frac{a}{s} > \bigg(\prod_{i = 1}^s \text{deg}(u_i)\bigg)^{s^{-1}},\nonumber
\end{eqnarray}
or, since both sides of the inequality are positive, equivalently  
\begin{eqnarray}
f(s) = \bigg(\frac{a}{s}\bigg)^s > \prod_{i = 1}^s \text{deg}(u_i).\nonumber
\end{eqnarray}
It remains to find an upper bound for the expression $f(s)$ when $s$ is between $1$ and $k(k+1)2\sqrt{h}$. The derivative 
$f'(s) = f(s) / \partial s$ of $f(s)$ with respect to $s$ is given by \begin{eqnarray}\label{eq:derivative}
f'(s) = \bigg(\frac{a}{s}\bigg)^s\cdot \bigg(\log\bigg(\frac{a}{s}\bigg)-1\bigg)\nonumber,
\end{eqnarray}
and since $f'(s)$ is positive for each $s \in [1,  k(k+1)2\sqrt{h}]$, it follows that $f(s)$ is maximized when $s=k(k+1)2\sqrt{h}$. Therefore, we obtain

\begin{eqnarray}\label{eq:upperbounded}
\left(2(k+1)\right)^{k(k+1)2\sqrt{h}} > \prod_{i = 1}^s \text{deg}(u_i), \nonumber 
\end{eqnarray}
finishing the proof of the claim.
\end{proof}

For the inductive step, suppose that $G$ contains more that $(k+1)2\sqrt{n}$ vertices of degree at most $k$ and that we can recolour any subgraph $H$ of $G$ with $h<n$ vertices to some $(k+1)$-colouring $\gamma^H$ such that each vertex get recoloured at most $k^{\bigO{k^2\sqrt{h}}}$ times.
Let $S$ be an independent set in $G$ containing only vertices of degree at most $k$ of size at least $2\sqrt{n}$. Since $G$ can be greedily coloured with $(k+1)$ colours using its $k$-degenerate ordering and $G$ contains more that $(k+1)2\sqrt{n}$ vertices of degree at most $k$, such a set $S$ exists and can be found in polynomial time. By the inductive hypothesis we can recolour the graph $H=G-S$ to some $(k+1)$-colouring such that each vertex get recoloured at most $k^{ck^2\sqrt{h}}$ times for some constant $c>1$. We can extend this sequence of recolourings to a sequence in $G$ by recolouring a vertex $u \in S$ 
whenever some neighbour of $u$ gets recoloured to its colour (this is possible because the number of colours is $k + 2$ and $u$ has at most $k$ neighbours in $G$). 
At the end of the sequence, we can recolour each vertex of $S$ to a colour other than $k + 2$. It follows that the maxmimum number $f(n)$ of times a vertex of $G$ is recoloured satisfies the inequality
\[
f(n) \le k \cdot \Big(k^{ck^2\sqrt{h}} \Big) + 1\le k \cdot \Big(k^{ck^2\sqrt{n - 2\sqrt{n}}} \Big) + 1 = k^{ck^2\sqrt{n - 2\sqrt{n}} + 1}+1.
\]
Since $k\geq 2$, to show that $f(n)\le k^{ck^2\sqrt{n}}$,
it suffice to show that $ ck^2\sqrt{n} > ck^2\sqrt{n - 2\sqrt{n}} + 1$ for each $n\ge 4$. Adding $-1$ to both sides of the inequality and then squaring yields the~result. 
\end{proof}

 In our next lemma, we adapt the proof method introduced in \cite{feghali} to show that we can further reduce the number of colours by one for planar graphs. 

\begin{lemma}\label{lem:master}
Let $G$ be a planar graph, and let $\gamma$ be a $6$-colouring of $G$. Then we can recolour $\gamma$ to some $5$-colouring of $G$ using seven colours by at most $\bigO{n^c}$ recolourings for some constant $c > 1$.\end{lemma}

\begin{proof}
Let $H$ be any subgraph of $G$, and let $h=|V(H)|$. 
An independent set $I$ of $H$ is said to be \emph{special} if it contains at least $h/49$ vertices and every vertex of $I$ has at most $6$ neighbours in $G - I$. 
Let $S$ be the set of vertices of $H$ of degree at most $6$. Then $S$ has at least $h/7$ vertices since otherwise
\[\sum_{v\in V(H)}\textrm{deg}(v) \geq \sum_{v\in V(H) - S}\textrm{deg}(v) > 7\bigg(h - \frac{h}{7}\bigg) = 6h,\]
which contradicts that $\textrm{mad}(G) < 6$. Let $I \subseteq S$ be a maximal independent subset of $S$. Each vertex of $I$ has at most $6$ neighbours in $S$ and every vertex of $S - I$ has at least one neighbour in $I$. Therefore, $|I| + 6|I| \geq |S|$ and so $I$ is a special independent set as needed.

Let us prove by induction on the order of $G$ that there is a sequence of recolourings from a $6$-colouring $\gamma$ of $G$ to 
some $5$-colouring of $G$. We will then argue that at most $\bigO{n^c}$ recolourings have been performed for some constant $c > 1$, thereby finishing the proof.

Let $I$ be a special independent set of $G$, and let $G^*$ be the graph obtained from $G$ by
\begin{itemize}
\item removing all vertices of degree $5$ in $I$ from $G$ and
\item for each vertex $v$ in $I$ of degree $6$, deleting $v$ and identifying a pair of neighbours of $v$ that are coloured alike in $\gamma$ (such a pair always exists since at most $6$ colours appear on $v$ and its neighbours).
\end{itemize}
Notice that $G^*$ is planar (one can think of some embedding of $G$ in the plane and then note that  the neighbours of any vertex $v$ form part of the boundary of a face $F$ in $G - v$; thus, indentifying a pair of neighbours of $v$ inside the interior of $F$ in the graph $G - v$ does not break the planarity).

Let $\gamma'$ denote the colouring of $G^*$ that agrees with $\gamma$ on $V(G^*) \cap  V(G)$ and such that, for each $z \in V(G^*) \setminus V(G)$, if $z$ is the vertex obtained by the identification of 
some vertices $x$ and $y$ of $G$, then $\gamma'(z) = \gamma(x) (=\gamma(y))$. Graph $G^*$ has less vertices than $G$, so can we apply our induction hypothesis to find a sequence of recolourings from $\gamma'$ 
to some $5$-colouring $\gamma''$ of~$G^*$. 

We let $\gamma^{\star}$ be the $5$-colouring of $G - I$ that agrees with $\gamma''$ on $V(G) \cap  V(G^*)$ and such that, for each pair of vertices $x, y \in V(G)$ identified into a new vertex $z$,
 $\gamma^{\star}(z) = \gamma''(x) =\gamma''(y)$. 
We can transform $\gamma'$ to $\gamma^{\star}$ by
\begin{itemize}
\item recolouring $x$ and $y$ using the same recolouring as $z$  for every pair $x, y \in V(G) - I$ identified into a vertex $z \in V(G^*)$;
\item recolouring each $v \in V(G^*) \cap V(G)$ using the same recolouring.
\end{itemize}
 
 We can extend this sequence to $G$ by recolouring each vertex of $I$ to a colour from $\{1, \dots, 7\}$ not appearing on it or its neighbours (this is possible since each vertex of $I$ either has degree at most $5$ or has degree $6$ but with at least two neighbours that are in some sense always coloured alike). At the end of this sequence, we recolour each vertex of $I$ of colour $7$ to another colour (this is again possible by the same reasoning). So our aim of transforming into a $5$-colouring is achieved unless some vertex of $I$ has colour $6$. 

Suppose that there is a vertex $v$ of $I$ with colour $6$. We emulate the proof of the $5$-Colour Theorem to show that we can recolour $v$ to a colour from $\{1, \dots, 5\}$ without introducing \emph{new} vertices of colour $6$ or $7$. By repeating the same procedure at most $|I|$ times, we can transform $\gamma$ into a $5$-colouring of $G$, as needed. For this, we require some definitions.    

Let $i$ and $j$ be two colours. Then a component $C$ of a subgraph of $G$ induced by colours $i$ and $j$ is called an $(i,j)$-component. Suppose that~$C$ is an $(i, j)$-component, $7 \notin \{i,j\}$. Then colours $i$ and $j$ are said to be \emph{swapped} on $C$ if the vertices coloured $j$ are recoloured $7$,  then the vertices coloured $i$ are recoloured $j$, and finally the vertices initially coloured $j$ are recoloured $i$. Since  no vertex coloured $i$ or $j$ in $C$ is adjacent to a vertex of colour $7$, it is clear that each colouring is proper and that no new vertices of colour $7$ are introduced. 

If for a vertex $v$ at least one colour in $\{1,\ldots,5\}$ does not appear on its neighbour, we can immediately recolour $v$.
So we can assume that $v$ has either degree $5$ or $6$ with precisely five neighbours $v_1, \dots, v_5$ coloured distinctly. Suppose these neighbours appear in this order in a plane embedding of $G$. Let us denote by $i$ the colour of $v_i$ ($i = 1, \dots, 5$). If the $(1, 3)$-component $C_{1, 3}$ that contains $v_1$ does not contain $v_3$, we swap colours $1$, $3$ on $C_{1, 3}$ (this is possible since colour $7$ is not used on $G$), which in turn allows us to recolour $v$ to $1$. So we can assume that $C_{1, 3}$ contains both $v_1$ and $v_3$. In the same vein, the vertices $v_2$ and $v_4$ must be contained in the same $(2,4)$-component $C_{2, 4}$. By the Jordan Curve Theorem, this is impossible. Hence, either $C_{1,3}$ does not contain both $v_1$ and $v_3$ or $C_{2,4}$ does not contain both $v_2$ and $v_4$ and we are able to recolour $v$, as required.  

Let us now estimate the number of recolourings of a vertex $v \in I$ in terms of the number of recolourings of vertices of $G - I$.  When recolouring~$\gamma$ to a $6$-colouring $\beta$ that uses only colours $1$ to $5$ on $G - I$, $v$ is recoloured at most five more times than any of its neighbours  (this bound is achieved if $v$ is recoloured every time one of its neighbours is recoloured and these neighbours are recoloured the same number of times). Moreover, recolouring $\beta$ to a $5$-colouring of $G$ has cost an additional $\bigO{n}$ recolourings per vertex. Therefore, the maximum number $f(n)$ of recolourings per vertex satisfies the recurrence relation
\[f(n) \leq 5 \cdot f\bigg(n - \frac{n}{49}\bigg) + \bigO{n},\]
and the theorem follows by the master theorem. 
\end{proof}

We also require some auxiliary results whose algorithmic versions (running in polynomial time) is implicit in the respective papers. 

\begin{lemma}[\cite{luisthesis}]\label{lem:ce}
Let $d$ and $k$ be positive integers, $k \geq 2d + 1$, and let $G$ be a $d$-degenerate graph on $n$ vertices. Then $R_k(G)$ has diameter $\bigO{n^2}$. 
\end{lemma}

\begin{lemma}[\cite{thomassen}]\label{lem:planarpartition}
Let $G = (V, E)$ be a planar graph. There is a partition $V = I \cup D$ such that $G[I]$ is an independent set and $G[D]$ is a $3$-degenerate graph.  
\end{lemma}

\begin{lemma}[\cite{mihok, wood}]\label{lem:degpartition}
Let $k$ be a positive integer, and let $G = (V, E)$ be a $k$-degenerate graph. There is a partition $V = I \cup F$ such that $G[I]$ is an independent set and $G[F]$ is  a $(k - 1)$-degenerate graph.
\end{lemma}

We combine Lemmas \ref{lem:planarpartition} and \ref{lem:degpartition} to obtain the following corollary. 

\begin{corollary}\label{cor:partition}
Let $G = (V, E)$ be a planar graph. Then there is a partition $V = I_1 \cup I_2 \cup A$ such that $G[I_1]$ and $G[I_2]$ are independent sets and $G[A]$ is a $2$-degenerate graph.
\end{corollary}

We are now ready to prove Theorem \ref{thm:planar}. 

\begin{proof}[Proof of Theorem \ref{thm:planar}]
Let $\alpha$ and $\beta$ be two $7$-colourings of $G$. To prove the theorem, it suffices to show that we can recolour $\alpha$ to $\beta$ by $2^{\bigO{\sqrt{n}}}$ recolourings. Combining Lemmas \ref{lem:sub} and \ref{lem:master}, we can recolour $\alpha$ to some $5$-colouring $\gamma_1$ of $G$ and $\beta$ to some $5$-colourings $\gamma_2$ by $2^{\bigO{\sqrt{n}}}$ recolourings. We apply Corollary \ref{cor:partition} to find a partition $V = I_1 \cup I_2 \cup A$ such that $G[I_1]$ and $G[I_2]$ are independent sets and $H = G[A]$ is a $2$-degenerate graph.
From $\gamma_1$ and $\gamma_2$ we recolour the vertices in $I_1$ to colour $7$ and those in $I_2$ to colour $6$ (the colours that are not used in either $\gamma_1$ or $\gamma_2$). Let $\gamma_1^H$ and $\gamma_2^H$ denote, respectively, the restrictions of $\gamma_1$ and $\gamma_2$ to~$H$. 
 We focus on $H$ and as long as we do not use colours $6$ and $7$ we can recolour $\gamma_1^H$ to $\gamma_2^H$ without worrying about adjacencies between $A$ and $I_1 \cup I_2$.   Since $H$ is $2$-degenerate, we can apply Lemma \ref{lem:ce} with $k = 5$ and $d = 2$ to find by $\bigO{n^2}$ recolourings a recolouring sequence from $\gamma_1^H$ to $\gamma_2^H$. This completes the proof of the theorem. \end{proof}
 
\section{Final remarks}

The reader may have observed from the proof of Theorem \ref{thm:planar} that, in order to settle \cite[Conjecture 16]{bonamy13},  it would suffice  to show that we can recolour any $7$-colouring of a planar graph to some $6$-colouring by polynomially many recolourings. 

\begin{problem}
Given a planar graph $G$ and a $7$-colouring $\alpha$ of $G$, can we recolour $\alpha$ to some $6$-colouring of $G$ by $O(n^c)$ recolourings for some constant $c > 0$?
\end{problem}

In order to obtain a sub-exponential bound on the diameter of reconfiguration graphs of colourings for graphs with any bounded maximum average degree, it would suffice to find a positive answer to the following problem. (The proof of this fact follows by combining Lemma \ref{lem:sub} with an affirmative answer to Problem \ref{prob:2} in the same way that Lemmas 8, 9 and 10 in \cite{feghalibrooks} are combined to obtain Theorem 6 in \cite{feghalibrooks}.)

\begin{problem}\label{prob:2}
Let $k\geq 2$, and let $G = (V, E)$ be a graph with $\textrm{mad}(G) < k$. Then there exists a partition $\{V_1, V_2\}$ of $G$ such that $G[V_1]$ is an independent set and $\textrm{mad}(G[V_2]) < k - 1$. 
\end{problem}

\section*{Acknowledgements}

Eduard Eiben was supported by Pareto-Optimal Parameterized Algorithms (ERC Starting Grant 715744).
Carl Feghali was supported by the research Council of Norway via the project CLASSIS.

 \bibliography{bibliography}{}
\bibliographystyle{abbrv}
 
\end{document}